\documentclass[letterpaper]{article} %
\usepackage{aaai2026}  %
\usepackage{times}  %
\usepackage{helvet}  %
\usepackage{courier}  %
\usepackage[hyphens]{url}  %
\usepackage{graphicx} %
\def\UrlFont{\rm}  %
\usepackage{natbib}  %
\usepackage{caption} %
\usepackage{algorithm}
\usepackage{algorithmic}
\usepackage{amssymb}

\usepackage{newfloat}
\usepackage{listings}
\DeclareCaptionStyle{ruled}{labelfont=normalfont,labelsep=colon,strut=off} %
\floatstyle{ruled}
\newfloat{listing}{tb}{lst}{}
\floatname{listing}{Listing}
\usepackage{amsmath}
\usepackage{graphicx}
\usepackage{multirow}
\usepackage{amsthm}
\usepackage{makecell}
\usepackage{enumitem}
\usepackage{algorithm}      %
\usepackage{algorithmic}    %
\usepackage{booktabs} 

\newtheorem{assumption}{Assumption}
\newtheorem{proposition}{Proposition}
\newtheorem{definition}{Definition}

\title{Towards Effective, Stealthy, and Persistent Backdoor Attacks \\Targeting Graph Foundation Models}
\author {
    Jiai Luo\textsuperscript{\rm 1},
    Qingyun Sun\textsuperscript{\rm 1},
    Lingjuan Lyu\textsuperscript{\rm 2},
    Ziwei Zhang\textsuperscript{\rm 1},
    Haonan Yuan\textsuperscript{\rm 1},\\
    Xingcheng Fu\textsuperscript{\rm 3},
    Jianxin Li\textsuperscript{\rm 1}\thanks{Corresponding Author}
}
\affiliations {
    \textsuperscript{\rm 1}SKLCCSE, School of Computer Science and Engineering, Beihang University, Beijing, China\\
    \textsuperscript{\rm 2}Sony AI, Zurich, Switzerland\\
    \textsuperscript{\rm 3}Key Lab of Education Blockchain and Intelligent Technology, Ministry of Education, Guangxi Normal University, China\\
    \{luojy, sunqy, zwzhang, yuanhn, lijx\}@buaa.edu.cn, 
    Lingjuan.Lv@sony.com, 
    fuxc@gxnu.edu.cn
}

\usepackage{bibentry}

\newcommand{\Model}{\textsc{Gfm-Ba}}

\begin{document}

\maketitle

\begin{abstract}
Graph Foundation Models (GFMs) are pre-trained on diverse source domains and adapted to unseen targets, enabling broad generalization for graph learning. Despite that GFMs have attracted considerable attention recently, their vulnerability to backdoor attacks remains largely underexplored. A compromised GFM can introduce backdoor behaviors into downstream applications, posing serious security risks. 
However, launching backdoor attacks against GFMs is non-trivial due to three key challenges.
(1) \textit{Effectiveness}: Attackers lack knowledge of the downstream task during pre-training, complicating the assurance that triggers reliably induce misclassifications into desired classes.
(2) \textit{Stealthiness}: The variability in node features across domains complicates trigger insertion that remains stealthy.
(3) \textit{Persistence}: Downstream fine-tuning may erase backdoor behaviors by updating model parameters. %
To address these challenges, we propose \Model\footnote{Code: \url{https://github.com/RingBDStack/GFM-BA}.}, a novel \underline{B}ackdoor \underline{A}ttack model against \underline{G}raph \underline{F}oundation \underline{M}odels.
Specifically, we first design a label-free trigger association module that links the trigger to a set of prototype embeddings, eliminating the need for knowledge about downstream tasks to perform backdoor injection. 
Then, we introduce a node-adaptive trigger generator, dynamically producing node-specific triggers, reducing the risk of trigger detection while reliably activating the backdoor. 
Lastly, we develop a persistent backdoor anchoring module that firmly anchors the backdoor to fine-tuning-insensitive parameters, enhancing the persistence of the backdoor under downstream adaptation.
Extensive experiments demonstrate the effectiveness, stealthiness, and persistence of \Model.%
\end{abstract}

\section{Introduction}
\label{sec:intro}
Graph Foundation Models (GFMs) are designed to be pre-trained on various graph data from diverse domains, and subsequently adapted to a wide range of downstream tasks in the target domain~\cite{liu2023towards,mao2024graph,shi2024graph,shi2024lecture}. 
Existing efforts~\cite{zhao2024all,yu2024text,yu2025samgpt,wang2024gft} towards GFMs have demonstrated strong knowledge transfer from pre-training source domains to target domains, achieving superior performance. 
While the pre-training and adaptation paradigm~\cite{zi2024prog,tang2024graphgpt,he2024unigraph,lachi2024graphfm} has driven the success of GFMs, they also introduce new potential security vulnerabilities, particularly backdoor attacks (i.e., inserting backdoors into the model that cause it to misbehave when encountering certain triggers). For GFMs, attackers can exploit the pre-training stage to inject backdoors and release compromised pre-trained GNNs to the public. Downstream users who adopt these pre-trained models unknowingly inherit the backdoor, exposing their downstream applications to targeted manipulation. These threats pose risks to critical applications of GFMs such as drug discovery~\cite{bongini2021molecular} and financial fraud detection~\cite{cheng2020graph}. 

\begin{figure*}[t]
  \centering
  \includegraphics[width=\textwidth]{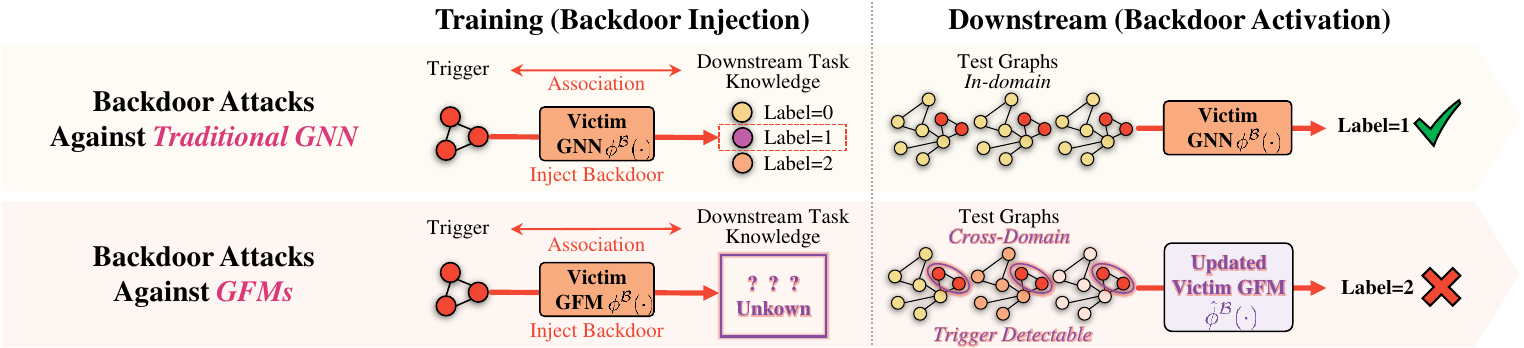}
  \vspace{-1.5em}
  \caption{Key differences between backdoor attacks against traditional GNNs and GFMs.}
  \vspace{-1em}
  \label{fig:intro}
\end{figure*}

Backdoor attacks for traditional GNNs have been extensively studied~\cite{zhang2021backdoor,xi2021graph,dai2023unnoticeable,zheng2023motif,xu2021explainability}. However, backdoor attacks against traditional GNNs and GFMs have fundamental differences. As shown in Figure~\ref{fig:intro}, existing backdoor attacks against GNNs operate under three presumed conditions where (1) labels for downstream tasks are accessible during the backdoor injection phase; (2) the training and downstream graphs originate from the same domain; (3) the backdoor model remains unchanged during the downstream usage. In the context of GFMs, all three conditions may not hold, \textit{leading to three key challenges for designing backdoor attacks against GFMs}: 
(1) \textit{Effectiveness}: During the pre-training stage, downstream task knowledge is inaccessible. How to ensure that the injected trigger consistently induces a specific label that aligns with the attacker's intent? 
(2) \textit{Stealthiness}: The distribution and semantics of node features can vary significantly across different graph domains~\cite{mao2024graph,shi2024graph}. 
How to design triggers that remain stealthy across diverse downstream domains? 
(3) \textit{Persistence}: Downstream adaptation may modify the learned model parameters, thereby erasing the backdoor effect (a phenomenon known as backdoor forgetting~\cite{gu2023gradient}). How can attackers embed backdoors that can remain persistently effective after downstream fine-tuning?

To address the aforementioned challenges, we propose \Model, a novel model for performing \underline{B}ackdoor \underline{A}ttacks against \underline{G}raph \underline{F}oundation \underline{M}odels. 
First, to solve the \textit{effectiveness} challenge and achieve label-specific manipulation without access to downstream knowledge, we design a label-free trigger association module, which links triggers to a set of prototype embeddings during pre-training. 
During the downstream trigger injection phase, the attacker can identify which prototype embedding aligns with the desired target label and injects the corresponding trigger through a few trial queries.
Second, to solve the \textit{Stealthiness} challenge and ensure that the injected trigger %
remains stealthy across diverse downstream domains, we develop a node-adaptive trigger generator, which dynamically produces context-aware triggers conditioned on the context of each target node instead of using fixed triggers. It enhances stealthiness while ensuring reliable backdoor activation. Besides, \Model\ does not modify the clean pre-trained model but instead activates the latent backdoor already hidden in the encoder. %
Lastly, to solve the \textit{persistence} challenge and maintain the functionality of the backdoor after downstream fine-tuning, we introduce a persistent backdoor anchoring module, which anchors the backdoor to parameters in the pre-trained model that are unlikely to change significantly during downstream adaptation. By embedding the backdoor into these stable regions of the model, the trigger-target association becomes less susceptible to being forgotten.

Our main contributions can be summarized as follows:
\begin{itemize}[leftmargin=*]
   \item We study backdoor attacks against Graph Foundation Models, highlighting the significant trustworthiness concerns in the development of GFMs. 
    \item We propose \Model, a novel backdoor attack model for GFMs containing three tailored modules targeting the \textit{effectiveness}, \textit{stealthiness}, and \textit{persistence} challenges.
    \item We conduct extensive experiments and show that \Model\ consistently outperforms existing methods against three representative victim GFMs, demonstrating its superior performance.
    
\end{itemize}

\section{Related Work}

\textbf{Graph Foundation Models}.
Graph Foundation Models (GFMs) aim to capture generalizable graph knowledge that enables positive transfer across tasks and domains~\cite{liu2023towards,mao2024graph,shi2024graph,wang2024gft,zi2024prog,tang2024graphgpt,he2024unigraph,lachi2024graphfm,xia2024opengraph}. They are typically pre-trained using self-supervised objectives, such as link prediction~\cite{yu2024text,zhang2018link} or graph contrastive learning~\cite{zhao2024all,yu2025samgpt}, over multiple source datasets. The resulting model is then adapted to downstream tasks on target graphs through task-specific fine-tuning~\cite{zhao2024all,hassani2022cross,you2020graph} or prompting~\cite{sun2022adversarial,sun2023all,sun2022gppt,fang2023universal,tang2024higpt}. For example,
GCOPE~\cite{zhao2024all} attempts to mitigate the negative transfer by introducing domain-specific virtual nodes that interconnect nodes across domains, aligning the semantic patterns. 
MDGPT~\cite{yu2024text} introduces a two-stage prompting strategy to adapt target domains by integrating unified multi-domain knowledge with domain-specific information. 
SAMGPT~\cite{yu2025samgpt} further introduces structure tokens to align various structural distributions.
However, the trustworthiness of GFMs remains largely unexplored in the current literature.

\textbf{Graph Backdoor Attacks}.
Graph backdoor attacks~\cite{lyu2024cross,zhang2023graph,zhang2021backdoor,xi2021graph,dai2023unnoticeable} aim to manipulate backdoored GNNs to predict a specified label for any input embedded with triggers (typically small subgraphs~\cite{zheng2023motif,xu2021explainability,wang2024multi,yang2024distributed,xu2022poster,yang2024graph,feng2024backdoor,luo2025robust}). 
Early efforts introduced subgraph-based triggers and achieved strong performance~\cite{zhang2021backdoor,xi2021graph}. Subsequent methods enhanced stealthiness by leveraging in-distribution triggers to evade outlier detection~\cite{dai2023unnoticeable,zhang2024rethinking}.
However, their focus on supervised GNNs makes them unsuitable for backdoor attacks on GFMs during pre-training.
GCBA~\cite{zhang2023graph} is the first backdoor attack targeting graph contrastive learning, but it requires knowledge of downstream class labels, limiting its applicability to GFMs. CrossBA~\cite{lyu2024cross} introduces the first cross-domain pre-training graph backdoor attack by aligning triggered graphs with a learned trigger embedding while separating them from clean samples. Yet it cannot control the predicted label, making the backdoor attack degenerate to a form of adversarial evasion attack~\cite{sun2022adversarial,kwon2025dual}. Additionally, the use of a fixed trigger across domains increases its detectability.
Overall, prior methods either rely on downstream knowledge or lack control over the attack outcome, limiting their effectiveness and stealthiness in the GFM setting.

\section{Threat Model}
\textbf{Attacker's Goals}: Consistent with prior work~\cite{lyu2024cross,zhang2023graph,xi2021graph}, we focus on the node classification task. In the context of backdoor attacks against GFMs, attackers aim to inject a backdoor into the pre-trained encoder $\phi(\cdot)$, resulting in a compromised model $\phi^{\mathcal{B}}(\cdot)$. The goal is to induce any downstream classifiers $f(\cdot)$ built on adapted $\hat{\phi}^{\mathcal{B}}(\cdot)$ to misclassify triggered graphs into a specific label $y^*$, while maintaining normal performance on clean inputs. Formally, the attack objective is:
\begin{equation}
f(\hat{\phi}^{\mathcal{B}}(\mathcal{G}^{\prime}, v)) = y^{*}, \quad
f(\hat{\phi}^{\mathcal{B}}(\mathcal{G}, v)) = f(\hat{\phi}(\mathcal{G}, v)),
\end{equation}
where $\mathcal{G}$ is a clean graph, $\mathcal{G}^{\prime}$ is the corresponding triggered graph, and $v$ denotes the target node.

\textbf{Attacker's Knowledge and Capabilities}: Following CrossBA~\cite{lyu2024cross}, attackers have full control over the pre-training process but have no access to any downstream information. 
\textit{\textbf{This setting is realistic, as foundation model training is resource-intensive, and downstream users often rely on open-source models or APIs to build their applications.}} \textit{Thus, a malicious adversary can release an open-source GFM, and downstream users who adopt it may unknowingly inherit the embedded backdoor.}
During the backdoor activation stage, the attacker can inject triggers into target nodes and observe feedback signals. 
For example, in paper ranking systems based on citation networks, attackers can publish dummy papers that cite the target paper. Public leaderboards then allow attackers to observe the feedback on the impact. 
In e-commerce, fake user accounts can manipulate and monitor recommendation outcomes.

\section{\Model}
In this section, we elaborate \Model, a novel model for \underline{B}ackdoor \underline{A}ttacks against \underline{G}raph \underline{F}oundation \underline{M}odels, designed to achieve effectiveness, stealthiness, and persistence simultaneously.
The overall architecture of \Model\ is shown in Figure~\ref{fig:frame}, which consists of three components:

(1) \textit{Label-Free Trigger Association Module (Challenge 1: Effectiveness)}. We apply Farthest Point Sampling to select prototype embeddings from the node embedding space of pre-training graphs generated by the pre-trained GNN encoder, which serve as targets for trigger association.

(2) \textit{Node-Adaptive Trigger Generator (Challenge 2: Stealthiness)}. We design a trigger generator to dynamically produce the personalized trigger for each target node. The generated trigger is optimized to align the triggered target node with the target embedding while maintaining graph homophily, ensuring both effectiveness and stealthiness.

(3) \textit{Persistent Backdoor Anchoring Module (Challenge 3: Persistence)}. We enhance backdoor persistence under downstream adaptation by anchoring the trigger-target association to fine-tuning-insensitive parameters of the pre-trained model. Specifically, we apply graph mixup to simulate distribution shifts and identify fine-tuning-sensitive parameters. Random perturbations are then applied to anchor the trigger-target mapping to the identified stable regions of the model.

\begin{figure*}[t]
  \centering
  \includegraphics[width=\textwidth]{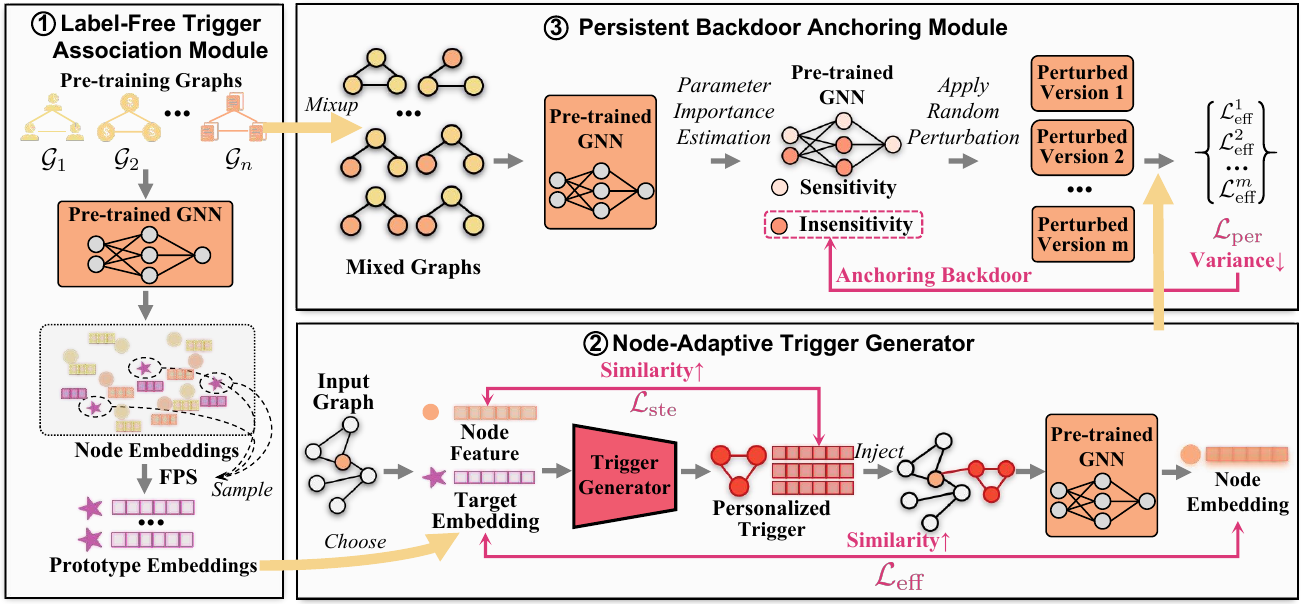}
  \vspace{-0.5em}
  \caption{The overall framework of \Model. FPS is first applied to select prototype embeddings as trigger association targets. The node-adaptive trigger generator then produces personalized triggers conditioned on the target embedding and the target node, ensuring both stealthiness and effectiveness. Finally, graph mixup is used to identify fine-tuning-insensitive parameters, allowing the backdoor to be anchored in the stable regions of the model and remain effective after the downstream adaptation.}
  \label{fig:frame}
  \vspace{-0.5em}
\end{figure*}

\subsection{Label-Free Trigger Association Module}
During pre-training, attackers lack access to the downstream model architecture and its decision boundaries. As a result, existing graph backdoor attacks that require the downstream task labels are infeasible in the context of GFMs~\cite{zhang2023graph, zhang2021backdoor, dai2023unnoticeable}. 
While methods such as CrossBA~\cite{lyu2024cross} avoid this limitation by pushing the triggered graph away from clean graphs and toward a learned trigger, they lack control over the resulting label, causing the triggered graph to be classified into a fixed but unknown class, failing to guide it to a specific target label.

To overcome this challenge, we propose associating triggers with a set of prototype embeddings rather than with specific labels. 
However, naively generating these prototype embeddings at random may result in all of them corresponding to a fixed and undesired class, whereas our goal is to construct prototype embeddings that span diverse downstream classes, enabling the attacker to steer the triggered graph toward a desired label by selecting an appropriate target embedding through several trial queries during the trigger injection phase. 
Following the assumption in prior work~\cite{zhao2024all, yu2024text, yu2025samgpt} that pre-training data contains knowledge relevant to downstream applications, sampling prototype embeddings from the node embedding space of pre-training graphs can effectively reduce the search space. Specifically, we employ the Farthest Point Sampling (FPS)~\cite{eldar1997farthest}, which is a greedy algorithm that iteratively selects the most distant point to ensure that the sampled set preserves the overall structure and distribution of the original data~\cite{eldar1997farthest,lang2020samplenet}. 
Here, we propose Proposition~\ref{theorem:fps}.

\begin{proposition}
When point density decays monotonically from each class centroid, increasing the separation between centroids raises the probability that FPS will cover more classes in a fixed number of steps.
\label{theorem:fps}
\end{proposition}

The detailed proof can be found in Appendix~\ref{appen:proof}. Based on Proposition~\ref{theorem:fps} and prior findings~\cite{chen2020simple, morcos2022understanding} that self-supervised pre-training enhances class separability, adopting FPS to select prototype embeddings from the node embeddings of the pre-training graphs encourages sampled embeddings to cover a broader range of downstream classes. 
\textit{Note that while ``density decays monotonically from class centroid'' is idealized and not strictly satisfied in practice, many real-world datasets exhibit this overall trend~\cite{idrissi2015improvement}, making proposition~\ref{theorem:fps} a useful theoretical motivation for using FPS.}
Specifically, let $\mathcal{E}$ denote the set of node embeddings extracted by the pre-trained GNN from the pre-training graphs. The target embedding set $\mathcal{E}_{\text{target}}$ is then constructed by selecting $k$ prototype embeddings from $\mathcal{E}$ using FPS, denoted as $\mathcal{E}_{\text{target}} = \text{FPS}(\mathcal{E}, k)$.

After associating triggers with $\mathcal{E}_{\text{target}}$, in the downstream trigger injection stage, attackers can select a target embedding from $\mathcal{E}_{\text{target}}$ based on the downstream application and generate its corresponding trigger using the trigger generator (see Section~\ref{sec:generator}) to activate the desired backdoor behavior.

\subsection{Node-Adaptive Trigger Generator} 
\label{sec:generator}
Graphs from different domains exhibit significant feature distribution discrepancy~\cite{shi2024graph}. Existing backdoor methods that use fixed, domain-agnostic triggers can introduce feature inconsistencies between the injected trigger node and its neighbors in the downstream graph, violating the graph homophily~\cite{jin2020graph}, thus making the injected trigger more detectable and easier to remove.

To avoid this problem, we replace the universal trigger with a trigger generator that dynamically produces a personalized trigger conditioned on both the target node and the selected target embedding. 
Specifically, we implement MLP as the trigger generator. Consistent with prior work~\cite{lyu2024cross}, we design the trigger as a fully connected 3-node graph with identical features, which is much smaller than both pre-training and downstream graphs. Let $\mathbf{x}_i$ denote the feature of the target node and $\mathbf{e}_{j}$ the target embedding chosen from $\mathcal{E}_{\text{target}}$. The feature of the injected trigger node, denoted as $\mathbf{x}_{ij}^{\text{tri}}$, is generated as follows:
\begin{equation} 
\mathbf{x}_{ij}^{\text{tri}} = \text{MLP}\left([\mathbf{x}_i \| \mathbf{e}_j]\right),
\end{equation}
where $[\cdot \, \| \, \cdot]$ denotes vector concatenation. We have two objectives for the trigger: (1) ensuring that the triggered graph is associated with the target embedding to guarantee effectiveness, and (2) maintaining feature similarity with the connected target node to enhance stealthiness. 
Let $\mathcal{G}_i$ denote a clean graph from the pre-training dataset $\{\mathcal{G}_{u}\}_{u=1}^{n}$, and l$\tilde{\mathcal{G}}_{ij}$ denote the 3-node trigger graph feature $\mathbf{x}_{ij}^{\text{tri}}$. We optimize the trigger generator with $\mathcal{L}_{\text{eff}}$ and $\mathcal{L}_{\text{ste}}$:
\begin{align}
  \mathcal{L}_{\text{eff}}\!
    &=\! -\,\mathbb{E}_{u\sim[n],\,\!j\sim[k],\,\!i\sim[|\mathcal G_u|]}\,\!
         \mathrm{sim}\left(\!\phi(\operatorname{In}(\mathcal G_u,\!\tilde{\mathcal G}_{ij},\!i)\!),\!e_j\!\right)\!\label{eq:loss_e},
  \\ 
  \mathcal{L}_{\text{ste}}\!
    &=\! -\,\mathbb{E}_{u\sim[n],\,\!j\sim[k],\,\!i\sim[|\mathcal G_u|]}\,\!
         \mathrm{sim}\left(x_i,x_{ij}^{\rm tri}\right)\label{eq:loss_s}.
\end{align}
where $|\mathcal{G}_{u}|$ denotes the number of nodes in $\mathcal{G}_{u}$, $\operatorname{In}(\mathcal{G}_u, \tilde{\mathcal{G}}_{ij}, i)$ denotes inserting the trigger graph $\tilde{\mathcal{G}}_{ij}$ into the clean graph $\mathcal{G}_u$ with the $i$-th node as the target node, $\phi(\cdot)$ is the frozen pre-trained GNN encoder, and $\operatorname{sim}(\cdot, \cdot)$ denotes the cosine similarity. Here, $\mathcal{L}_{\text{eff}}$ promotes the association between the embedding of the triggered graph and the target embedding, while $\mathcal{L}_{\text{ste}}$ encourages the feature similarity between the trigger node and the connected target node. 
Notably, the backdoor vulnerability is introduced without the need to modify the parameters of the pre-trained GNN, but instead leverages the inherent latent backdoor logic already present in the encoder, enabling the GFMs to preserve their performance on clean graphs without any degradation.

\begin{figure}[t]
  \centering
  \includegraphics[width=0.8\linewidth]{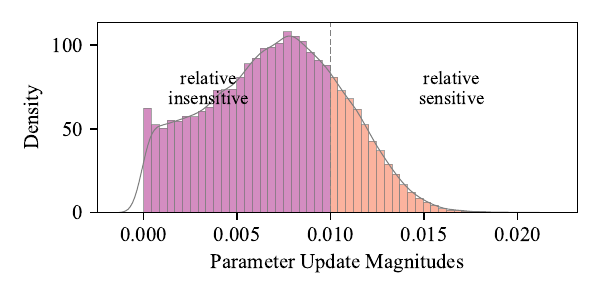}
  \vspace{-0.5em}
  \caption{An analysis of the distribution of parameter update magnitudes after downstream fine-tuning.}
  \vspace{-0.5em}
  \label{fig:pc}
\end{figure}
\subsection{Persistent Backdoor Anchoring Module}
During the downstream model construction, downstream users may fine-tune the pre-trained GNN, and such parameter updates can invalidate the backdoor, known as backdoor forgetting~\cite{gu2023gradient}.

To mitigate this negative effect, we draw inspiration from the backdoor attacks in NLP~\cite{cheng2023backdoor} where triggers are crafted using rare tokens to minimize the likelihood of related parameters being updated during downstream fine-tuning. Our key idea is to anchor the association between the trigger and the target embedding to parameters of the pre-trained GNN that are insensitive to fine-tuning in downstream applications. 
To investigate this possibility, we conducted a preliminary experiment in which we measured the magnitude of parameter updates of a pre-trained GCN~\cite{kipf2016semi} during downstream fine-tuning, with the details can be found in Appendix~\ref{appen:exp}. 
As shown in Figure~\ref{fig:pc}, most parameters are insensitive (i.e., change slightly after fine-tuning), while only a few are significantly updated. This motivates us to anchor the trigger–target association to these insensitive parameters.

\begin{proposition}
    Consider a GNN with parameters $\boldsymbol{\Theta}$. Let $\mathbf{X}$ be the node features, $\mathbf{A}$ be the adjacency matrix, and $\mathbf{Z}$ be the output. Assume the Jacobian matrix $J_{\partial \mathbf{Z}/\partial \theta_k}(\mathbf{X})$ is rank-deficient. For any direction $\Delta \boldsymbol{\Theta} \in \bigcup_{\mathbf{X}} \ker J_{\partial \mathbf{Z}/\partial \theta_k}(\mathbf{X})$, there exists an input $\mathbf{X}$ such that $\mathbf{Z}$ is first-order insensitive to $\Delta \boldsymbol{\Theta}$, where $\ker$ denotes the null space of the Jacobian matrix, and $\bigcup$ denotes the union operator.
\label{theorem:insensitive}
\end{proposition}
The proof is provided in Appendix~\ref{appen:proof}. Proposition~\ref{theorem:insensitive} suggests the possibility that the association between the trigger and the target embedding can be anchored to insensitive parameters of GNNs. 
To achieve this goal, we first apply graph mixup~\cite{ling2023graph} to explore potential downstream graph patterns.
We then adopt the parameter importance estimation method from model pruning~\cite{molchanov2019importance}, leveraging the pre-training loss on mixed graphs to identify insensitive parameters.
Finally, we introduce random perturbations to the identified sensitive parameters during the training process of the trigger generator, guiding the generator to anchor the trigger-target link to the insensitive parameters.
Specifically, for each $(\mathcal{G}_{i}\! = \!\{\mathbf{X}_{i}, \mathbf{A}_{i}\}, \mathcal{G}_{j}\! =\! \{\mathbf{X}_{j}, \mathbf{A}_{j}\})$ in the pre-training dataset with $n$ graphs, we synthesize a mixed graph $\mathcal{G}^{\text{mix}}\!=\!\{\mathbf{X}^{\text{mix}}\!,\! \mathbf{A}^{\text{mix}}\}$ as follows:
\begin{equation}\label{eq:mixed}
\begin{aligned}
  \mathbf{A}^{\mathrm{mix}}
    &= \lambda \mathbf{A}_{i} + (1 - \lambda)\,\mathbf{M}\,\mathbf{A}_{j}\,\mathbf{M}^{\top},\\
  \mathbf{X}^{\mathrm{mix}}
    &= \lambda \mathbf{X}_{i} + (1 - \lambda)\,\mathbf{M}\,\mathbf{X}_{j}.
\end{aligned}
\end{equation}
where $\lambda$ is the hyperparameter, $\mathbf{M}$ is the alignment matrix calculated as $\mathbf{M} = \text{softmax}(\text{sim}(\mathbf{R}_i, \mathbf{R}_j))$, with $\mathbf{R}_i$ and $\mathbf{R}_j$ denoting node representations of $\mathcal{G}_{i}$ and $\mathcal{G}_{j}$ computed via two rounds of message-passing.

We treat the mixed graph set $\{\mathcal{G}^{\text{mix}}_i\}_{i=1}^{n(n-1)}$ as a proxy for potential downstream patterns, based on the widely accepted assumption that pre-training graphs encode knowledge relevant to downstream applications~\cite{mao2024graph,shi2024graph,zhao2024all}. 
To estimate parameters sensitive to fine-tuning, we adopt the importance-based sensitivity measure from model pruning~\cite{molchanov2019importance}. Given the pre-training loss $\mathcal{L}_{\text{pre}}$ evaluated on the mixed graph set, the sensitivity of parameter $\theta_k$ is defined as:
\begin{equation}
\label{eq:sensitive}
    \mathcal{I}(\theta_{k})=\left(g_{k}\theta_{k}-\frac{1}{2}\theta_{k}\mathbf{H}_{k}\boldsymbol{\Theta}\right)^{2},
\end{equation}
where $g_{k}=\partial\mathcal{L}_{\text{pre}}/\partial\theta_{k}$ and $\mathbf{H}_{k}$ represents the $k$-th row of the corresponding Hessian matrix.
A larger value of $\mathcal{I}(\theta_{k})$ implies that the parameter $\theta_{k}$ is more sensitive. We then introduce random perturbations to the top $s$\% most sensitive ones. 
For each selected sensitive parameter $\theta_{k}$, the perturbation is applied as $\theta_{k}\leftarrow\theta_{k}+ \epsilon|\theta_{k}|$, where $\epsilon$ is sampled from a Gaussian distribution $\mathcal{N}(0,\sigma^{2})$.
We perform $m$ perturbations and obtain a set of perturbed parameters $\{\boldsymbol{\Theta}^{(i)}\}_{i=1}^{m}$ of the pre-trained GNN.
These perturbed parameters will yield a set of effective loss $\mathcal{L}_{\text{eff}}$ as defined in Eq.~\eqref{eq:loss_e}, denoted as $\{\mathcal{L}_{\text{eff}}^{j}\}_{j=1}^{m}$.
To mitigate the degradation of backdoor effectiveness caused by sensitive parameters during downstream adaptation, we introduce a persistence loss $\mathcal{L}_{\text{per}}$:
\begin{equation} 
\label{eq:loss_p}
\mathcal{L}_{\text{per}} = \text{Var}(\{\mathcal{L}_{\text{eff}}^{j}\}_{j=1}^{m}) +\text{Mean}(\{\mathcal{L}_{\text{eff}}^{j}\}_{j=1}^{m}).
\end{equation}
Minimizing $\mathcal{L}_{\text{per}}$ enhances the persistence of trigger-target association under downstream fine-tuning.

The overall pipeline of \Model\ is in Appendix~\ref{appen:alg}.

\section{Experiments}
\label{sec:exp}

\begin{table*}[t]
  \centering
  
  \setlength{\tabcolsep}{0.8\tabcolsep}
  \resizebox{\textwidth}{!}{%
    \begin{tabular}{c|c|c c|c c|c c|c c|c c}
      \toprule
      \multicolumn{2}{c|}{\textbf{Dataset}} & \multicolumn{2}{c|}{\textbf{Cora}} & \multicolumn{2}{c|}{\textbf{CiteSeer}} & \multicolumn{2}{c|}{\textbf{PubMed}} & \multicolumn{2}{c|}{\textbf{Photo}} & \multicolumn{2}{c}{\textbf{Computers}} \\
      \midrule
      \textbf{Victim} & \textbf{Threaten} & \makecell{\textbf{ASR}\\(\textit{Scen.1})} & \makecell{\textbf{ASR}\\(\textit{Scen.2})} & \makecell{\textbf{ASR}\\(\textit{Scen.1})} & \makecell{\textbf{ASR}\\(\textit{Scen.2})} & \makecell{\textbf{ASR}\\(\textit{Scen.1})} & \makecell{\textbf{ASR}\\(\textit{Scen.2})} & \makecell{\textbf{ASR}\\(\textit{Scen.1})} & \makecell{\textbf{ASR}\\(\textit{Scen.2})} & \makecell{\textbf{ASR}\\(\textit{Scen.1})} & \makecell{\textbf{ASR}\\(\textit{Scen.2})} \\
      \midrule
      \multirow{4}[2]{*}{GCOPE} & GCBA\_R & 26.78\scalebox{0.5}{±8.89} & 3.83\scalebox{0.5}{±1.27} & 29.64\scalebox{0.5}{±8.67} & 4.94\scalebox{0.5}{±1.45} & 62.10\scalebox{0.5}{±12.81} & 20.70\scalebox{0.5}{±4.27} & 26.80\scalebox{0.5}{±8.79} & 3.35\scalebox{0.5}{±1.10} & 35.40\scalebox{0.5}{±19.13} & 3.54\scalebox{0.5}{±1.91} \\
          & GCBA\_M & 33.42\scalebox{0.5}{±1.89} & 4.77\scalebox{0.5}{±0.27} & 35.87\scalebox{0.5}{±6.50} & 5.98\scalebox{0.5}{±1.08} & 64.94\scalebox{0.5}{±15.18} & 21.65\scalebox{0.5}{±5.06} & 27.80\scalebox{0.5}{±9.26} & 3.48\scalebox{0.5}{±1.16} & 46.20\scalebox{0.5}{±12.34} & 4.62\scalebox{0.5}{±1.23} \\
          & CrossBA & \textbf{100.00\scalebox{0.5}{±0.00}} & \underline{14.29\scalebox{0.5}{±0.00}} & \textbf{100.00\scalebox{0.5}{±0.00}} & \underline{16.67\scalebox{0.5}{±0.00}} & \textbf{100.00\scalebox{0.5}{±0.00}} & \underline{33.33\scalebox{0.5}{±0.00}} & \underline{74.00\scalebox{0.5}{±13.44}} & \underline{9.25\scalebox{0.5}{±1.68}} & \underline{79.80\scalebox{0.5}{±7.26}} & \underline{7.98\scalebox{0.5}{±0.73}} \\
          & \textbf{\Model} & \textbf{100.00\scalebox{0.5}{±0.00}} & \textbf{90.40\scalebox{0.5}{±12.10}} & \textbf{100.00\scalebox{0.5}{±0.00}} & \textbf{89.06\scalebox{0.5}{±8.02}} & \textbf{100.00\scalebox{0.5}{±0.00}} & \textbf{100.00\scalebox{0.5}{±0.00}} & \textbf{93.20\scalebox{0.5}{±7.05}} & \textbf{84.53\scalebox{0.5}{±1.79}} & \textbf{94.40\scalebox{0.5}{±7.70}} & \textbf{78.54\scalebox{0.5}{±7.28}} \\
      \midrule
      \multirow{4}[2]{*}{SAMGPT} & GCBA\_R & 61.66\scalebox{0.5}{±8.82} & 8.81\scalebox{0.5}{±1.26} & 51.60\scalebox{0.5}{±13.57} & 8.60\scalebox{0.5}{±2.26} & 89.02\scalebox{0.5}{±18.39} & 29.67\scalebox{0.5}{±6.13} & 52.00\scalebox{0.5}{±9.95} & 6.50\scalebox{0.5}{±1.24} & 66.60\scalebox{0.5}{±6.66} & 6.66\scalebox{0.5}{±0.67} \\
          & GCBA\_M & 61.44\scalebox{0.5}{±13.41} & 8.78\scalebox{0.5}{±1.92} & 64.42\scalebox{0.5}{±15.49} & 10.74\scalebox{0.5}{±2.58} & 94.82\scalebox{0.5}{±4.64} & 31.61\scalebox{0.5}{±1.55} & 57.40\scalebox{0.5}{±9.58} & 7.18\scalebox{0.5}{±1.20} & 68.00\scalebox{0.5}{±12.83} & 6.80\scalebox{0.5}{±1.28} \\
          & CrossBA &\underline{ 95.24\scalebox{0.5}{±10.64}} & \underline{13.61\scalebox{0.5}{±1.52}} & \textbf{100.00\scalebox{0.5}{±0.00}} & \underline{16.67\scalebox{0.5}{±0.00}} & \textbf{100.00\scalebox{0.5}{±0.00}} & \underline{33.33\scalebox{0.5}{±0.00}} & \underline{96.80\scalebox{0.5}{±4.55}} & \underline{12.10\scalebox{0.5}{±0.57}} & \underline{92.00\scalebox{0.5}{±11.75}} & \underline{9.20\scalebox{0.5}{±1.17}} \\
          & \textbf{\Model} & \textbf{100.00\scalebox{0.5}{±0.00}} & \textbf{100.00\scalebox{0.5}{±0.00}} & \textbf{100.00\scalebox{0.5}{±0.00}} & \textbf{100.00\scalebox{0.5}{±0.00}} & \textbf{100.00\scalebox{0.5}{±0.00}} & \textbf{100.00\scalebox{0.5}{±0.00}} & \textbf{100.00\scalebox{0.5}{±0.00}} & \textbf{99.80\scalebox{0.5}{±0.23}} & \textbf{100.00\scalebox{0.5}{±0.00}} & \textbf{100.00\scalebox{0.5}{±0.00}} \\
      \midrule
      \multirow{4}[2]{*}{MDGPT} & GCBA\_R & -     & -     & -     & -     & -     & -     & -     & -     & -     & - \\
          & GCBA\_M & -     & -     & -     & -     & -     & -     & -     & -     & -     & - \\
          & CrossBA & \underline{95.14\scalebox{0.5}{±5.26}} & \underline{13.59\scalebox{0.5}{±0.75}} & \underline{92.36\scalebox{0.5}{±15.87}} & \underline{15.39\scalebox{0.5}{±2.65}} & \textbf{100.00\scalebox{0.5}{±0.00}} & \underline{33.33\scalebox{0.5}{±0.00}} & \underline{82.60\scalebox{0.5}{±19.86}} & \underline{10.32\scalebox{0.5}{±2.48}} & \underline{82.00\scalebox{0.5}{±19.21}} & \underline{8.20\scalebox{0.5}{±1.92}} \\
          & \textbf{\Model} & \textbf{100.00\scalebox{0.5}{±0.00}} & \textbf{96.61\scalebox{0.5}{±6.17}} & \textbf{100.00\scalebox{0.5}{±0.00}} & \textbf{99.43\scalebox{0.5}{±0.95}} & \textbf{100.00\scalebox{0.5}{±0.00}} & \textbf{100.00\scalebox{0.5}{±0.00}} & \textbf{98.40\scalebox{0.5}{±2.30}} & \textbf{97.68\scalebox{0.5}{±2.57}} & \textbf{99.20\scalebox{0.5}{±1.30}} & \textbf{93.19\scalebox{0.5}{±3.66}} \\
      \bottomrule
    \end{tabular}%
  }
  \vspace{-0.5em}
  \caption{The results of attack effectiveness. \textit{Scen.1} refers to \textit{target-uncontrolled attack}, and \textit{Scen.2} refers to \textit{target-controlled attack}. ``-'' indicates that the method is not applicable to perform the backdoor attack (GCBA~\cite{zhang2023graph} is tailored for the graph contrastive learning method). The best results are shown in \textbf{bold} and the runner-ups are \underline{underlined}.}
  \label{tab:effectiveness}
\end{table*}

\begin{table*}[t]
  \centering
  \setlength{\tabcolsep}{0.8\tabcolsep}
  \resizebox{\textwidth}{!}{%
    \begin{tabular}{c|c|c c|c c|c c|c c|c c}
      \toprule
      \multicolumn{2}{c|}{\textbf{Dataset}} & \multicolumn{2}{c|}{\textbf{Cora}} & \multicolumn{2}{c|}{\textbf{CiteSeer}} & \multicolumn{2}{c|}{\textbf{PubMed}} & \multicolumn{2}{c|}{\textbf{Photo}} & \multicolumn{2}{c}{\textbf{Computers}} \\
      \midrule
      \textbf{Victim} & \textbf{Threaten} & \makecell{\textbf{ACC}\\(\textit{Clean})} & \makecell{\textbf{ASR}\\(\textit{Purified})} & \makecell{\textbf{ACC}\\(\textit{Clean})} & \makecell{\textbf{ASR}\\(\textit{Purified})} & \makecell{\textbf{ACC}\\(\textit{Clean})} & \makecell{\textbf{ASR}\\(\textit{Purified})} & \makecell{\textbf{ACC}\\(\textit{Clean})} & \makecell{\textbf{ASR}\\(\textit{Purified})} & \makecell{\textbf{ACC}\\(\textit{Clean})} & \makecell{\textbf{ASR}\\(\textit{Purified})} \\
      \midrule
      \multirow{4}[2]{*}{GCOPE} & GCBA\_R & 59.18\scalebox{0.5}{±3.93} & 20.87\scalebox{0.5}{±9.57} & 56.30\scalebox{0.5}{±2.86} & 27.28\scalebox{0.5}{±12.83} & 55.00\scalebox{0.5}{±4.39} & 42.18\scalebox{0.5}{±13.59} & 57.00\scalebox{0.5}{±3.24} & 26.60\scalebox{0.5}{±9.56} & 44.60\scalebox{0.5}{±3.71} & 35.00\scalebox{0.5}{±20.04} \\
          & GCBA\_M & 59.78\scalebox{0.5}{±7.20} & 23.27\scalebox{0.5}{±8.89} & 56.12\scalebox{0.5}{±2.13} & 22.62\scalebox{0.5}{±6.51} & 51.10\scalebox{0.5}{±5.79} & 48.71\scalebox{0.5}{±22.34} & 58.80\scalebox{0.5}{±7.33} & 26.80\scalebox{0.5}{±8.90} & 47.20\scalebox{0.5}{±7.26} & 43.20\scalebox{0.5}{±13.55} \\
          & CrossBA & \underline{60.52\scalebox{0.5}{±1.68}} & \underline{52.04\scalebox{0.5}{±2.42}} & \underline{59.06\scalebox{0.5}{±4.45}} & \underline{57.58\scalebox{0.5}{±2.80}} & \underline{51.36\scalebox{0.5}{±5.37}} & \underline{90.65\scalebox{0.5}{±2.52}} & \underline{65.60\scalebox{0.5}{±3.51}} & \underline{67.20\scalebox{0.5}{±9.73}} & \underline{50.40\scalebox{0.5}{±4.83}} & \underline{48.80\scalebox{0.5}{±15.85}} \\
          & \textbf{\Model} & \textbf{61.46\scalebox{0.5}{±3.25}} & \textbf{100.00\scalebox{0.5}{±0.00}} & \textbf{60.10\scalebox{0.5}{±5.72}} & \textbf{100.00\scalebox{0.5}{±0.00}} & \textbf{54.44\scalebox{0.5}{±4.43}} & \textbf{100.00\scalebox{0.5}{±0.00}} & \textbf{65.80\scalebox{0.5}{±4.97}} & \textbf{100.00\scalebox{0.5}{±0.00}} & \textbf{54.80\scalebox{0.5}{±7.56}} & \textbf{100.00\scalebox{0.5}{±0.00}} \\
      \midrule
      \multirow{4}[2]{*}{SAMGPT} & GCBA\_R & 60.36\scalebox{0.5}{±1.89} & 38.24\scalebox{0.5}{±6.24} & 44.10\scalebox{0.5}{±4.47} & 43.92\scalebox{0.5}{±10.22} & 58.92\scalebox{0.5}{±4.35} & 77.80\scalebox{0.5}{±17.21} & 79.20\scalebox{0.5}{±4.32} & 17.80\scalebox{0.5}{±3.11} & 69.00\scalebox{0.5}{±6.00} & 26.20\scalebox{0.5}{±4.15} \\
          & GCBA\_M & 61.94\scalebox{0.5}{±3.56} & 29.32\scalebox{0.5}{±5.45} & 48.76\scalebox{0.5}{±4.09} & 48.38\scalebox{0.5}{±12.30} & 63.02\scalebox{0.5}{±2.60} & 86.26\scalebox{0.5}{±9.05} & 76.20\scalebox{0.5}{±5.67} & 18.20\scalebox{0.5}{±3.11} & \textbf{71.20\scalebox{0.5}{±2.59}} & 23.20\scalebox{0.5}{±4.82} \\
          & CrossBA & \underline{62.82\scalebox{0.5}{±2.85}} & \underline{70.28\scalebox{0.5}{±11.22}} & \underline{59.04\scalebox{0.5}{±2.86}} & \underline{75.90\scalebox{0.5}{±1.95}} & \underline{65.64\scalebox{0.5}{±4.40}} & \underline{74.08\scalebox{0.5}{±10.75}} & \underline{80.40\scalebox{0.5}{±4.34}} & \underline{67.80\scalebox{0.5}{±2.39}} & 67.20\scalebox{0.5}{±7.12} & \underline{50.60\scalebox{0.5}{±6.95}} \\
          & \textbf{\Model} & \textbf{63.54\scalebox{0.5}{±3.73}} & \textbf{88.28\scalebox{0.5}{±3.08}} & \textbf{61.72\scalebox{0.5}{±3.61}} & \textbf{86.04\scalebox{0.5}{±1.56}} & \textbf{65.92\scalebox{0.5}{±5.65}} & \textbf{93.62\scalebox{0.5}{±2.20}} & \textbf{80.60\scalebox{0.5}{±3.58}} & \textbf{86.00\scalebox{0.5}{±2.12}} & \underline{69.20\scalebox{0.5}{±4.09}} & \textbf{84.60\scalebox{0.5}{±2.19}} \\
      \midrule
      \multirow{4}[2]{*}{MDGPT} & GCBA\_R & -     & -     & -     & -     & -     & -     & -     & -     & -     & - \\
          & GCBA\_M & -     & -     & -     & -     & -     & -     & -     & -     & -     & - \\
          & CrossBA & \underline{42.36\scalebox{0.5}{±7.08}} & \underline{47.76\scalebox{0.5}{±13.84}} & \underline{37.82\scalebox{0.5}{±5.00}} & \underline{61.12\scalebox{0.5}{±20.64}} & \underline{50.20\scalebox{0.5}{±6.87}} & \underline{69.10\scalebox{0.5}{±20.81}} & \underline{68.20\scalebox{0.5}{±9.09}} & \underline{35.40\scalebox{0.5}{±5.73}} & \underline{50.20\scalebox{0.5}{±6.10}} & \underline{37.00\scalebox{0.5}{±8.69}} \\
          & \textbf{\Model} & \textbf{60.88\scalebox{0.5}{±4.83}} & \textbf{81.30\scalebox{0.5}{±3.60}} & \textbf{60.58\scalebox{0.5}{±2.82}} & \textbf{85.78\scalebox{0.5}{±2.96}} & \textbf{62.48\scalebox{0.5}{±4.71}} & \textbf{92.36\scalebox{0.5}{±1.97}} & \textbf{79.20\scalebox{0.5}{±6.14}} & \textbf{89.60\scalebox{0.5}{±5.68}} & \textbf{71.80\scalebox{0.5}{±3.49}} & \textbf{85.20\scalebox{0.5}{±4.82}} \\
      \bottomrule
    \end{tabular}%
  }
  \vspace{-0.5em}
  \caption{Results of attack stealthiness. ACC reports the accuracy of the backdoored model on clean, non-triggered input graphs. \textit{Purified} refers to applying edge-based purification to the triggered graphs under Scenario 1. }
  
  \label{tab:stealthiness}
\end{table*}

\subsection{Experiments Settings}
\textbf{Datasets}. We evaluate \Model\ on five widely used node-level classification datasets for GFMs. Cora~\cite{yang2016revisiting}, CiteSeer~\cite{yang2016revisiting}, and PubMed~\cite{yang2016revisiting} are citation networks where nodes correspond to scientific publications and edges represent citation links. Computers~\cite{shchur2018pitfalls} and Photos~\cite{shchur2018pitfalls} are the Amazon co-purchase dataset, where edges indicate frequent co-purchase relationships between products. 
For dataset splitting, we use the standard PyG~\cite{fey2019fast} splits for Cora, CiteSeer, and PubMed. Following~\cite{yu2024text}, we designate the last 100 nodes as the test set for Computers and Photo, and randomly sample $m$ labeled nodes from the remaining nodes to form the $m$-shot training set.

\textbf{Victim GFMs}. We evaluate the backdoor performance of \Model\ against three state-of-the-art GFMs designed for multi-domain graph pre-training and adaptation: (1) GCOPE~\cite{zhao2024all} employs virtual nodes to interconnect graphs across domains, aligning the semantics of graphs. (2) MDGPT\cite{yu2024text} introduces dual prompts to adapt to target domains while integrating unified multi-domain knowledge with a tailored mixture of domain-specific prompts. (3) SAMGPT~\cite{yu2025samgpt} incorporates structural tokens to unify multi-domain structural knowledge and adapt it effectively to unseen domains.

\textbf{Baselines}. We adopt CrossBA~\cite{lyu2024cross}, the state-of-the-art graph backdoor method against cross-domain graph pre-training, as our primary baseline. Following~\cite{lyu2024cross}, we also include two adapted variants of GCBA~\cite{zhang2023graph}: GCBA\_R, which randomly selects a cluster center as the target embedding, and GCBA\_M, which selects the most isolated cluster center.

\textbf{Implement Details.}
Following prior works~\cite{zhao2024all,yu2024text,yu2025samgpt}, GFMs are trained across all datasets except the one held out for testing.
All experiments are conducted in a 5-shot node classification setting. A grid search over $\alpha$ and $\beta$ is performed in the range $[1e^{-2}, 1e^0]$ using a logarithmic step size of 5. Each experiment is repeated 5 times on a single NVIDIA V100 GPU. 
Full hyperparameter configurations and more details are in Appendix~\ref{appen:exp}.

\textbf{Evaluation Metrics.}
Following~\cite{lyu2024cross}, we adopt two evaluation metrics: (1) Attack Success Rate (ASR) measures the proportion of triggered samples misclassified into the target class. (2) Accuracy (ACC) measures performance on clean test data.%

\subsection{Attack Effectiveness Evaluation} 
\label{sec:exp_effect}
To evaluate the effectiveness of \Model, we consider two attack scenarios based on whether the attacker can specify the target label.
\textbf{Scenario 1: Target-Uncontrolled Attack}: The model is expected to misclassify triggered samples consistently into an arbitrary but fixed target class. 
\textbf{Scenario 2: Target-Controlled Attack}. The model is expected to misclassify triggered samples into the attacker-specified target class, aligned with the attacker's intent.

As shown in Table~\ref{tab:effectiveness}, \Model\ consistently achieves the highest ASR across all datasets in both attack scenarios, demonstrating its strong attack capability. 
In Scenario 2, the improvements are particularly notable, with gains ranging from 66.67\% to 90.80\% over the baselines across five datasets.
This is because baseline methods associate the trigger with an unknown fixed target label and cannot align it with the attacker-specified label. Consequently, they degenerate into targeted adversarial evasion attacks~\cite{sun2022adversarial,kwon2025dual} and fail to satisfy the requirements of a targeted backdoor attack in Scenario 2.

\subsection{Attack Stealthiness Evaluation}
To ensure the practical utility of backdoor attacks, stealthiness is a critical requirement that involves two key aspects: (1) the backdoored model should maintain normal behavior on clean input graphs; and (2) the inserted trigger should be inconspicuous and resistant to detection. 
To assess the first aspect, we measure the ACC of the backdoored GFMs on the clean test set.
To assess the second aspect, we apply a simple purification strategy that removes edges between node pairs with feature cosine similarity below 0.1 and report the ASR of the purified triggered graphs under Scenario 1.

As shown in Table~\ref{tab:stealthiness}, \Model\ maintains high clean accuracy for its ability to leverage latent backdoor logic in the pre-trained GNN without the need to modify model parameters. 
Moreover, \Model\ consistently achieves the highest ASR across all target GFMs and datasets after graph purification, outperforming baselines by an average of 36.81\%, 19.98\%, and 36.73\% on GCOPE, MDGPT, and SAMGPT, respectively, demonstrating its strong stealthiness. 

\begin{table}[t]  %
  \centering

  \setlength{\tabcolsep}{0.5\tabcolsep}
  \resizebox{1.0\linewidth}{!}{%
    \begin{tabular}{c|c|cc|cc|cc}
    \toprule
    \multicolumn{2}{c|}{\textbf{Dataset}} & \multicolumn{2}{c|}{\textbf{Cora}} & \multicolumn{2}{c|}{\textbf{Photo}} & \multicolumn{2}{c}{\textbf{Computers}} \\
    \midrule
    \multicolumn{1}{c|}{\textbf{Victim}} & \multicolumn{1}{c|}{\textbf{Threaten}} & \textbf{ASR}   & \textit{Drop}  & \textbf{ASR}   & \textit{Drop} & \textbf{ASR}   & \textit{Drop} \\
    \midrule
    \multirow{4}[2]{*}{SAMGPT} & GCBA\_R & 52.22\scalebox{0.5}{±15.68} & \textit{$\downarrow$9.44} & 38.80\scalebox{0.5}{±14.99} & \textit{$\downarrow$13.20} & 28.80\scalebox{0.5}{±3.96} & \textit{$\downarrow$6.60} \\
          & GCBA\_M & 54.58\scalebox{0.5}{±9.34} & \textit{$\downarrow$6.86} & 44.40\scalebox{0.5}{±21.11} & \textit{$\downarrow$13.00} & 40.60\scalebox{0.5}{±10.45} & \textit{$\downarrow$5.60} \\
          & CrossBA & \underline{90.50\scalebox{0.5}{±13.68}} & \underline{\textit{$\downarrow$4.74}} & \underline{87.40\scalebox{0.5}{±12.64}} & \underline{\textit{$\downarrow$9.40}} & \underline{91.40\scalebox{0.5}{±10.50}} & \textbf{\textit{$\downarrow$0.60}} \\
          & \textbf{\Model} & \textbf{98.66\scalebox{0.5}{±1.85}} & \textbf{\textit{$\downarrow$1.34}} & \textbf{96.00\scalebox{0.5}{±2.55}} & \textbf{\textit{$\downarrow$4.00}}  & \textbf{98.60\scalebox{0.5}{±1.95}} & \underline{\textit{$\downarrow$1.40}} \\
    \midrule
    \multirow{4}[2]{*}{MDGPT} & GCBA\_R & -     & -     & -     & -     & -     & - \\
          & GCBA\_M & -     & -     & -     & -     & -     & - \\
          & CrossBA & \underline{93.78\scalebox{0.5}{±8.05}} & \underline{\textit{$\downarrow$1.36}} & \underline{78.00\scalebox{0.5}{±13.38}} & \underline{\textit{$\downarrow$4.60}} & \underline{79.60\scalebox{0.5}{±16.95}} & \underline{\textit{$\downarrow$2.40}} \\
          & \textbf{\Model} & \textbf{99.32\scalebox{0.5}{±1.31}} & \textbf{\textit{$\downarrow$0.68}} & \textbf{97.80\scalebox{0.5}{±1.64}} & \textbf{\textit{$\downarrow$0.60}} & \textbf{98.40\scalebox{0.5}{±2.19}} & \textbf{\textit{$\downarrow$0.80}} \\
    \bottomrule
    \end{tabular}%
    
  }
  \caption{Results of attack persistence under Scenario 1. \textit{Drop} indicates the decrease in ASR compared to the setting without fine-tuning the pre-trained encoder.}
  \vspace{-0.5em}
  \label{tab:persistenceness}
\end{table}

\subsection{Attack Persistence Evaluation}
The effectiveness of a backdoor attack typically depends on the parameters of the pre-trained model. However, in practical scenarios, downstream users may fine-tune the pre-trained GNN encoder, which can break the association between the trigger and the target embedding.
To evaluate the persistence of \Model, we conduct experiments on the Cora, Photo, and Computers datasets, targeting SAMGPT and MDGPT as victim models. The pre-trained GNN encoder is fine-tuned using a learning rate of 0.001, and we report the ASR under Scenario 1 described in Section~\ref{sec:exp_effect}, keeping all other experimental settings unchanged. 
As shown in Table~\ref{tab:stealthiness}, baseline methods exhibit a performance drop after fine-tuning the pre-trained encoder, whereas \Model\ maintains a high ASR. This suggests that the backdoor logic of our \Model\ is deeply embedded and resistant to forgetting during downstream adaptation.

\subsection{Ablation Study}

To evaluate the contribution of our designed label-free trigger association module, node-adaptive trigger generator, and persistent trigger anchoring module in \Model, we conduct ablation studies using three variants:
(1) \Model(w/o E): We randomly set a target embedding under the target-controlled scenario to assess the impact of removing the label-free trigger association module.
(2) \Model(w/o S): A static trigger is used for all nodes under the graph purification defense in the target-uncontrolled scenario to evaluate the role of the node-adaptive trigger generator.
(3) \Model(w/o P): The persistence loss is removed, and the pre-trained encoder is fine-tuned to examine persistence in the absence of the persistent trigger anchoring module. 
Results on the Photo and Computers datasets with SAMGPT serving as the victim GFM are shown in Figure~\ref{fig:ablation}, while other datasets and victim models show similar patterns. As observed, the three components contribute to the effectiveness, stealthiness, and persistence of attacks, respectively.

\begin{figure}[t]
  \centering
  \includegraphics[width=1.0\linewidth]{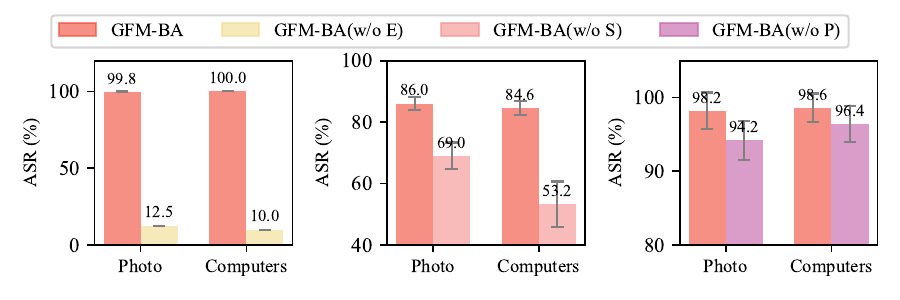}
  \vspace{-2em}
  \caption{Results of ablation studies on Photo and Computers. \Model(w/o E) is evaluated in the target-controlled scenario to assess effectiveness without the label-free trigger association module. \Model(w/o S) is tested in the target-uncontrolled scenario with graph purification to assess stealthiness. \Model(w/o P) is tested in the target-uncontrolled scenario with a fine-tuned backbone to assess the backdoor persistence.}
  \label{fig:ablation}
\end{figure}

\subsection{Hyperparameter Study}
\begin{figure}[t]
  \centering
  \vspace{-0.5em}
  \includegraphics[width=1.0\linewidth]{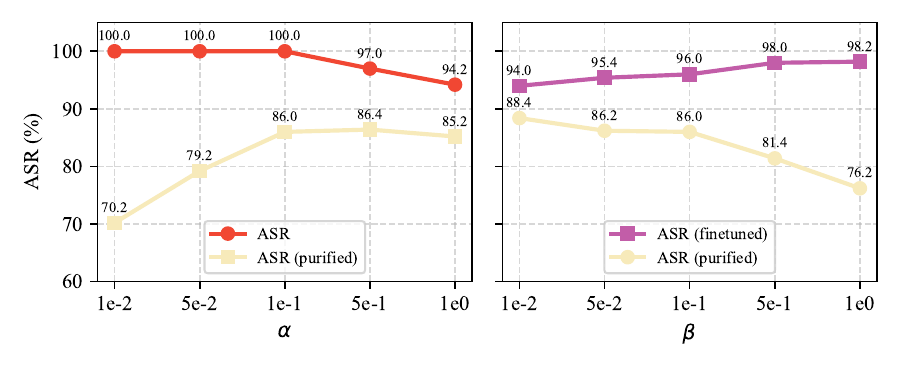}
  \vspace{-2.5em}
  \caption{Hyperparameter study on the Photo dataset.}
  \label{fig:hyper}
\end{figure}

In this section, we analyze the sensitivity of \Model\ to hyperparameters $\alpha$ and $\beta$. For $\alpha$, we examine its impact on the ASR of both the original triggered graphs and the purified triggered graphs, reflecting the trade-off between attack effectiveness and stealthiness. For $\beta$, we evaluate the ASR after fine-tuning the pre-trained encoder as well as under purification, assessing its influence on the persistence of the backdoor. Results on the Photo dataset under the target-uncontrolled scenario with SAMGPT as the victim are shown in Figure~\ref{fig:hyper}.

\section{Conclusion}
\label{sec:con}
We propose \Model, a novel backdoor attack model against Graph Foundation Models during the pre-training stage, designed to ensure effectiveness, stealthiness, and persistence. To perform effective backdoor when downstream tasks are unknown during pre-training, we introduce a label-free trigger association module that associates the trigger with a set of prototype embeddings selected via FPS. To enhance stealthiness, we design a node-adaptive trigger generator that produces trigger features close to the target node. For persistence, we anchor the trigger to parameters that are less sensitive to fine-tuning. Extensive experiments against representative victim GFMs demonstrate that \Model\ consistently outperforms existing methods across all objectives. 

\clearpage
\section*{Acknowledgements}
The corresponding author is Jianxin Li.
This work was supported by the National Natural Science Foundation of China under Grants No. 62225202 and No. 62302023, and by the State Key Laboratory of Complex \& Critical Software Environment (CCSE-2024ZX).
We express our sincere gratitude to all reviewers for their valuable efforts and contributions.
\bibliography{reference}
\clearpage
\clearpage
\appendix

\setcounter{secnumdepth}{2}

\setcounter{section}{0}
\renewcommand{\thesection}{\Alph{section}}

\setcounter{theorem}{0}
\renewcommand{\thetheorem}{\Alph{section}\arabic{theorem}}

\setcounter{assumption}{0}
\renewcommand{\theassumption}{\Alph{section}\arabic{assumption}}

\setcounter{proposition}{0}
\renewcommand{\theproposition}{\Alph{section}\arabic{proposition}}

\setcounter{definition}{0}
\renewcommand{\thedefinition}{\Alph{section}\arabic{definition}}

\section{Proof}
\label{appen:proof}

\subsection{Proof of Proposition~\ref{theorem:fps}}

We first restate the theorem as follows:
\begin{proposition}
When point density decays monotonically from each class centroid, increasing the separation between centroids raises the probability that FPS will cover more classes in a fixed number of steps.
\end{proposition}

\begin{proof}
Given $m$ classes. 
For class $i\in[m]$ we observe $n_i$ i.i.d. points $X_{i,\ell}=c_i+Z_{i,\ell}\in\mathbb R^d, \ell=1,\dots,n_i,$ where $c_i\in\mathbb R^d$ is the class centroid and $Z_{i,\ell}$ are i.i.d., centered at the origin.
In this proposition, we restrict our attention to scenarios where the sampling density decays monotonically from each class centroid. We begin by formally stating the underlying assumptions:

\begin{assumption}
For each class $i$, the class density is pointwise nonincreasing in radius about $c_i$: for all $r_1<r_2$ and any unit vectors at all directions $u,v\in\mathbb S^{d-1}$, $f_i(c_i+r_1u)\ \ge\ f_i(c_i+r_2v).$ Equivalently, the radial variable $R_i=\|Z_{i,1}\|$ has a nonincreasing density on $[0,\infty)$
\end{assumption}

\begin{definition}
\textbf{Farthest-Point Sampling (FPS).} Given a finite dataset $\mathcal X\subset\mathbb R^d$, FPS picks an arbitrary seed $x_1\in\mathcal X$; at step $t\ge2$ it selects $x_t\in\arg\max_{x\in\mathcal X}\operatorname{dist}\!\bigl(x,\{x_1,\dots,x_{t-1}\}\bigr), \operatorname{dist}(x,S):=\min_{p\in S}\|x-p\|.$ Let $C_k$ be the number of distinct class labels among $\{x_1,\dots,x_k\}$.
\end{definition}

Let $C_k$ be the number of distinct class labels among $\{x_1,\dots,x_k\}$. We formalize the phrase ``increase the separation between centroids'' with the following structural, geometry-agnostic condition. 
We compare datasets indexed by a nonnegative parameter $\lambda\ \ge\ 1$ that controls how far apart class centroids are. For each $\lambda$, let $\mathcal X(\lambda)$ be the dataset formed with the same noise draws $\{Z_{i,\ell}\}$ but with centroids changed to $c_i(\lambda)$. We only require:

\begin{assumption}
Consider a family of datasets $\{\mathcal X(\lambda)\}_{\lambda\ge1}$ built from the same per-class relative positions $\{Z_{i,\ell}\}$ but with centroids moved farther apart (the concrete motion is arbitrary). We require: \\(i) For any two samples from the \emph{same} class, their pairwise distance is independent of $\lambda$. \\ (ii) For any two samples from \emph{different} classes, their pairwise distance is nondecreasing in $\lambda$.
\end{assumption}

In words, increasing separation does not change intra-class geometry and never makes cross-class pairs closer.
Write $\Delta(\lambda):=\min_{i\neq j}\|c_i(\lambda)-c_j(\lambda)\|$ for the minimum inter-centroid distance at separation level $\lambda$.
For $\alpha\in(0,1)$, define the class-wise $\alpha$-quantile radius as follows:
\begin{align}
    &\rho_i(\alpha):=\inf\{r\ge0:\ \mathbb P(R_i\le r)\ge\alpha\},\\
&\rho_{\max}(\alpha):=\max_{i\le m}\rho_i(\alpha).
\end{align}
Define event $\mathcal E_\alpha:=\bigcap_{i=1}^m\bigcap_{\ell=1}^{n_i}\bigl\{\ \|X_{i,\ell}-c_i\|\le\rho_i(\alpha)\ \bigr\}.$
Because samples within a class are i.i.d., $\mathbb P(\mathcal E_\alpha)\ge \alpha^{\,N}$.

Denote $S_{t-1}$ the set of points that FPS has already selected in the first $t-1$ iterations.We begin by introducing two simple facts that will facilitate the analysis later:\\
(i) If a class is covered, then for any of its samples $x$ there exists $p\in S_{t-1}$ from the same class, and $\|x-p\|\le \|x-c_i\|+\|p-c_i\|\le 2\,\rho_{\max}(\alpha).$ Hence, for any covered-class sample $x$, $\operatorname{dist}(x,S_{t-1})\ \le\ 2\,\rho_{\max}(\alpha).$\\
(ii) If class $i$ is uncovered, then every $p\in S_{t-1}$ is from some class $j\neq i$. Therefore, $\|x-p\|\ \ge\ \|c_i(\lambda)-c_j(\lambda)\|-\|x-c_i(\lambda)\|-\|p-c_j(\lambda)\| \ge\ \Delta(\lambda)-2\,\rho_{\max}(\alpha),$ so for any uncovered-class sample $x$,$\operatorname{dist}(x,S_{t-1})\ \ge \Delta(\lambda)-2\,\rho_{\max}(\alpha).$ \\
Note that these bounds depend only on $\rho_{\max}(\alpha)$ and $\Delta(\lambda)$, not on the specific points already chosen.

We compare two separation levels $\lambda_1\le \lambda_2$. Run FPS on the same underlying sample at both levels, producing sequences $x_t(\lambda_1)$ and $x_t(\lambda_2)$ with selected sets $S_{t-1}(\lambda_1)$ and $S_{t-1}(\lambda_2)$. Define event $\mathcal A_t(\lambda) $ as follows:
\begin{equation}
\mathcal A_t(\lambda) := 
\left\{
  \parbox{0.65\linewidth}{
    the \(t\)-th FPS selection at level \(\lambda\) comes from a previously \emph{uncovered} class
  }
\right\}
\end{equation}

On the event $\mathcal{E}_\alpha$ (the two aforementioned inequalities, facts (i) and (ii)) imply that: \\
(1) Every uncovered candidate at level $\lambda_2$ is at distance at least $\Delta(\lambda_2)-2\rho_{\max}(\alpha)$ from its selected set. \\
(2) Every covered candidate at any level is at a distance at most $2\rho_{\max}(\alpha)$ from its selected set.

Because $\lambda_2\ge\lambda_1$ and $\Delta(\lambda)$ is nondecreasing in $\lambda$, we have:
\begin{equation}
\Delta(\lambda_2)-2\rho_{\max}(\alpha)\ge\Delta(\lambda_1)-2\rho_{\max}(\alpha).
\end{equation}
Hence, if step $t$ selects an uncovered class at level $\lambda_1$, then the uncovered-vs-covered gap at level $\lambda_2$ is no smaller, and step $t$ also selects an uncovered class at level $\lambda_2$. 
In symbols, we have:
\begin{equation}
\mathbb P\bigl(\mathcal A_t(\lambda_2)\cap \mathcal E_\alpha\bigr)
\;\ge\;
\mathbb P\bigl(\mathcal A_t(\lambda_1)\cap \mathcal E_\alpha\bigr).
\end{equation}
Use $\mathbb P(B)\ge \mathbb P(B\cap G)=\mathbb P(B)-\mathbb P(B\cap G^c)\ge \mathbb P(B)-\mathbb P(G^c)$ with $B=\mathcal A_t(\lambda_i)$, $G=\mathcal E_\alpha$, we have:
\begin{equation}
\begin{aligned}
\mathbb P\bigl(\mathcal A_t(\lambda_2)\bigr)
&\ \ge\
\mathbb P\bigl(\mathcal A_t(\lambda_1)\bigr)
\ -\
\bigl(1-\mathbb P(\mathcal E_\alpha)\bigr) \\
&\ \ge\
\mathbb P\bigl(\mathcal A_t(\lambda_1)\bigr)
\ -\
\bigl(1-\alpha^{\,N}\bigr).
\end{aligned}
\end{equation}
Letting $\alpha\uparrow 1$ (so that $1-\alpha^{\,N}\downarrow 0$) yields $\mathbb P\bigl(\mathcal A_t(\lambda_2)\bigr)\ \ge\ \mathbb P\bigl(\mathcal A_t(\lambda_1)\bigr).$
Finally, the number of classes covered in the first $k$ steps is as follows:
\begin{equation}
    C_k(\lambda)=\sum_{t=1}^k \mathbf 1_{\mathcal A_t(\lambda)}.
\end{equation}
Summing the inequalities over $t=1,\dots,k$ gives, for every $r\in\{1,\dots,\min\{k,m\}\}$, we have:
\begin{equation}
\mathbb P\{C_k(\lambda_2)\ge r\}\ \ge\ \mathbb P\{C_k(\lambda_1)\ge r\}.
\end{equation}
This is exactly the proposition: increasing centroid separation (larger $\lambda$) increases the probability that FPS samples from at least $r$ classes within $k$ steps.

This completes the proof.
\end{proof}

\subsection{Proof of Proposition~\ref{theorem:insensitive}}

We first restate the theorem as follows:

\begin{proposition}
Consider a GNN with parameters $\boldsymbol{\Theta}$. Let $\mathbf{X}$ be the node features, $\mathbf{A}$ be the adjacency matrix, and $\mathbf{Z}$ be the output. Assume the Jacobian matrix $J_{\partial \mathbf{Z}/\partial \theta_k}(\mathbf{X})$ is rank-deficient. For any direction $\Delta \boldsymbol{\Theta} \in \bigcup_{\mathbf{X}} \ker J_{\partial \mathbf{Z}/\partial \theta_k}(\mathbf{X})$, there exists an input $\mathbf{X}$ such that $\mathbf{Z}$ is first-order insensitive to $\Delta \boldsymbol{\Theta}$, where $\ker$ denotes the null space of the Jacobian matrix, and $\bigcup$ denotes the union operator.
\end{proposition}

\begin{proof}

We first formally state the following assumption:

\begin{assumption}
We assume the Jacobian matrix $J_{\partial \mathbf{Z} / \partial \theta_k}(\mathbf{X})$ is rank-deficient.
\end{assumption}

This assumption is generally valid in practice, as most neural networks are over-parameterized~\cite{zou2019improved}.

Let us define the sensitivity function
\begin{equation}
G(\mathbf{X}) := \frac{\partial \mathbf{Z}}{\partial \theta_k},
\end{equation}
which captures how the output \( \mathbf{Z} \) changes with respect to a parameter \( \theta_k \in \boldsymbol{\Theta} \). The Jacobian \( J_G(\mathbf{X}) \) describes the sensitivity of this function to parameter perturbations.

Suppose that for some input \( \mathbf{X} \), the Jacobian \( J_G(\mathbf{X}) \) is rank-deficient. Then, by definition of the kernel, there exists a nonzero direction \( \Delta \boldsymbol{\Theta} \in \ker J_G(\mathbf{X}) \) such that
\begin{equation}
J_G(\mathbf{X}) \Delta \boldsymbol{\Theta} = 0.
\end{equation}
This implies that the directional derivative of \( \mathbf{Z} \) along \( \Delta \boldsymbol{\Theta} \) vanishes:
\begin{equation}
\frac{\partial \mathbf{Z}}{\partial \theta_k} \cdot \Delta \boldsymbol{\Theta} = 0.
\end{equation}

Therefore, the output is first-order insensitive to perturbations in \( \theta_k \) along this direction.

Since \( J_G(\mathbf{X}) \) is rank-deficient, the rank-nullity theorem ensures that the kernel is non-trivial, and such directions always exist. Consequently, for every \( \Delta \boldsymbol{\Theta} \in \bigcup_{\mathbf{X}} \ker J_{\partial \mathbf{Z}/\partial \theta_k}(\mathbf{X}) \), there exists at least one input \( \mathbf{X} \) such that:
\begin{equation}
\frac{\partial \mathbf{Z}^{(L)}}{\partial \theta_k} \cdot \Delta \boldsymbol{\Theta} = 0.
\end{equation}

This completes the proof.
\end{proof}

\section{Algorithm}\label{appen:alg}
The overall training pipeline of \Model\ is detailed in Algorithm~\ref{alg:training}.

\renewcommand{\algorithmicrequire}{\textbf{Input:}}  
\renewcommand{\algorithmicensure}{\textbf{Output:}}  

\begin{algorithm}
\caption{The overall training pipeline of \Model.}
\begin{algorithmic}
\label{alg:training}
\REQUIRE Pre-trained GNN $\phi(\cdot)$, pre-training graphs $\{\mathcal{G}_{i}\}_{i=1}^{n}$, hyperparameters $\alpha,\beta$.
\ENSURE Trigger generator.
\STATE Sample $\mathcal{E}_{\text{target}}$ from node embeddings $\mathcal{E}$ obtained by the pre-trained GNN using FPS.
\STATE Construct the mixed graph set using pre-training graphs as described in Eq.~\ref{eq:mixed}.
\STATE Identify sensitive and insensitive parameters as defined in Eq.~\ref{eq:sensitive}.
\FOR{$\text{epoch} = 1,\dots,t$}
    \FOR{each node $n_{i}$ in the pre-training graphs}
        \STATE Construct the induced graph and randomly sample a target embedding $\mathbf{e}_i$ from $\mathcal{E}_{\text{target}}$.
        \STATE Compute $\mathcal{L}_{\text{eff}}$, $\mathcal{L}_{\text{ste}}$, and $\mathcal{L}_{\text{per}}$ using Eq.~\ref{eq:loss_e}, Eq.~\ref{eq:loss_s}, and Eq.~\ref{eq:loss_p}, respectively.
        \STATE Optimize the trigger generator using $\mathcal{L} = \mathcal{L}_{\text{eff}} + \alpha \mathcal{L}_{\text{ste}} + \beta \mathcal{L}_{\text{per}}$.
    \ENDFOR
\ENDFOR
\end{algorithmic}
\end{algorithm}

\section{Experiments Details}
\label{appen:exp}

\subsection{Experimental Details for Figure~\ref{fig:pc}}
We conduct a preliminary experiment to examine the distribution of parameter update magnitudes in a pre-trained GNN after downstream fine-tuning. The results are shown in Figure~\ref{fig:pc}. 
The pre-training datasets are CiteSeer, PubMed, Computers, and Photos, while the downstream test dataset is Cora, as described in Section~\ref{sec:exp}. 
We adopt a two-layer GCN as the GNN encoder. The learning rate for pre-training is set to 0.0001, with a maximum of 10,000 epochs and an early stopping patience of 20. For fine-tuning the pre-trained encoder, the learning rate is set to 0.001, and epochs is 500.
\begin{table}[t]
  \centering
  \caption{Hyperparameter settings across different datasets and victim GFMs.}
  \setlength{\tabcolsep}{0.5\tabcolsep}
  \resizebox{0.5\linewidth}{!}{%
    \begin{tabular}{c|c|ccc}
    \toprule
    \textbf{Victim} & \textbf{Dataset} & $t$ & $\alpha$ & $\beta$ \\
    \midrule
    \multirow{5}[2]{*}{GCOPE} & Cora  & 5     & 0.1   & 0.1 \\
          & CiteSeer & 5     & 0.1   & 0.1 \\
          & PubMed & 5     & 0.1   & 0.1 \\
          & Photo & 5     & 0.1   & 0.1 \\
          & Computers & 5     & 0.1   & 0.1 \\
    \midrule
    \multirow{5}[2]{*}{MDGPT} & Cora  & 1     & 0.1   & 0.1 \\
          & CiteSeer & 1     & 0.1   & 0.1 \\
          & PubMed & 1     & 0.1   & 0.1 \\
          & Photo & 1     & 0.1   & 0.1 \\
          & Computers & 1     & 0.1   & 0.1 \\
    \midrule
    \multirow{5}[2]{*}{SAMGPT} & Cora  & 1     & 0.1   & 0.1 \\
          & CiteSeer & 1     & 0.1   & 0.1 \\
          & PubMed & 1     & 0.1   & 0.1 \\
          & Photo & 1     & 0.1   & 0.1 \\
          & Computers & 1     & 0.1   & 0.1 \\
    \bottomrule
    \end{tabular}%
  }
  \label{tab:hyper-setting}
\end{table}

\subsection{Implement Details}
During the training stage of \Model, we follow the procedure in~\cite{zhao2024all, sun2023all, lyu2024cross} by splitting the pre-training graphs into induced subgraphs, with the smallest and largest subgraph sizes set to 15 and 30, respectively. 
For all the datasets, we set the number of prototype embeddings to $k=500$ (only 1.12\% $\sim$ 1.82\% of the total number of nodes in the pre-training datasets), the standard deviation of the Gausssian distribution used to sample the perturbation noise $\epsilon$ to $\sigma=0.1$, the perturbation ratio to $s=0.2$, the mixing ratio for synthesizing the mixed graph to $\lambda = 0.5$, and the number of perturbation iterations to $m=5$. The hidden dimension of the MLP-based trigger generator is set to 128. The remaining hyperparameters for each dataset across the three victim GFMs are listed in Table~\ref{tab:hyper-setting}. All datasets are available on Pyg~\cite{fey2019fast} and consented by the authors for academic usage.  

\subsection{Experiments compute resources}
All experiments are conducted on a single NVIDIA V100 GPU, with approximately 20 GB of GPU memory usage in our experiments. The training time for \Model\ ranges from approximately 15 $\sim$ 30 minutes across the five datasets against three representative victim GFMs.

\section{Further Discussion.}
\label{appen:limit}
One limitation of this work is its focus on static graphs, leaving the exploration of dynamic, temporal, and heterogeneous graph structures as an important direction for future research. 
This work highlights a previously largely underexplored vulnerability of GFMs during pre-training, aiming to raise awareness and promote research on robust, trustworthy graph learning. By exposing these risks, we hope to facilitate the development of more secure GFMs and support their responsible deployment in real-world settings.

\end{document}